\documentclass{llncs}

\RequirePackage[backend=biber,bibencoding=utf8,sorting=anyt,style=numeric,maxbibnames=6,giveninits=true]{biblatex}
\usepackage[unicode=true, colorlinks=true, citecolor=black, linkcolor=black, urlcolor=black]{hyperref}
\usepackage{xurl}
\hypersetup{breaklinks=true}

\usepackage[english=british]{csquotes}
\usepackage{mathtools}
\mathtoolsset{showonlyrefs}

\usepackage{algorithm}
\usepackage{algpseudocode}

\usepackage{cryptocode}

\usepackage[capitalise,nameinlink,noabbrev]{cleveref}

\crefname{claim}{Claim}{Claims}
\Crefname{claim}{Claim}{Claims}

\usepackage{etoolbox} 

\usepackage{tikz}
\usetikzlibrary{shapes,arrows.meta,positioning,quotes,calc}
\tikzstyle{block} = [draw, rectangle split, rectangle split parts=2, minimum
height=1cm, text width=2.8cm]
\tikzstyle{virtual} = [coordinate]
\tikzstyle{treenode} = [circle, align=center, inner sep=0pt, text centered,
draw=black, text width=.8cm]
\tikzstyle{leaf} = [rectangle, align=center, text centered,
draw=black, thick, minimum height=.8cm]
\tikzstyle{s-treenode} = [treenode, draw=red, ultra thick]

\usepackage{todonotes}
\usepackage{calc,ifthen}
\newcounter{hours}
\newcounter{minutes}
\newcommand{\Printtime}{\setcounter{hours}{\time/60}%
  \setcounter{minutes}{\time-\value{hours}*60}%
  \thehours:%
  \ifthenelse{\value{minutes}<10}{0}{}\theminutes}

\usepackage{enumitem}
\newlist{myclaim}{enumerate}{3}
\setlist[myclaim,1]{label=\arabic*.,
  ref =\thelemma.\arabic*}
\crefname{myclaimi}{Claim}{Claims}

\newcommand{\concat}{\mathbin\Vert}

\newcommand{\ShortLong}[2]{#2}

\title{Quotable Signatures for \\Authenticating Shared Quotes\thanks{The first and third authors were supported in part by the Independent Research Fund Denmark, Natural Sciences, grant DFF-0135-00018B.
    All authors are currently associated with DDC -- the Digital Democracy Center at the University of Southern Denmark.}
}

\renewcommand{\orcidID}[1]{}
\author{Joan Boyar\inst{1}\orcidID{0000-0002-0725-8341} \and
  Simon Erfurth\inst{1}\orcidID{0000-0001-8862-2856} \and
  Kim S. Larsen\inst{1}\orcidID{0000-0003-0560-3794} \and
  Ruben Niederhagen\inst{1,2}}

\institute{University of Southern Denmark, Odense, Denmark\\
  \email{joan@imada.sdu.dk}\\
  \email{simon@serfurth.dk}\\
  \email{kslarsen@imada.sdu.dk} \and
        Academia Sinica, Taipei, Taiwan\\
        \email{ruben@polycephaly.org}}

\date{\today{}}
\begin{document}

\maketitle
\begin{abstract}
  Quotable signature schemes are digital signature schemes with the additional
  property that from the signature for a message, any party can extract signatures
  for (allowable) quotes from the message, without knowing the secret key or
  interacting with the signer of the original message.
  Crucially, the extracted signatures are still signed with the original secret
  key.
  We define a notion of security for quotable signature schemes and
  construct a concrete example of a quotable signature scheme, using Merkle
  trees and classical digital signature schemes.
  The scheme is shown to be secure, with respect to the aforementioned notion of
  security.
  Additionally, we prove bounds on the complexity of the constructed
  scheme\ShortLong{.}{ and provide algorithms for signing, quoting, and verifying.}
  Finally, concrete use cases of quotable signatures are considered, using them to
  combat misinformation by bolstering authentic content on social media.
  We consider both how quotable signatures can be used, and why using them
  could help mitigate the effects of fake news.

  \keywords{quotable signatures \and
    digital signatures \and
    Merkle trees \and
    authenticity \and
    fake news}
\end{abstract}




\section{Introduction}
\label{sec:intro}
Digital signature schemes are a classical and widely used tool in modern
cryptography (the canonical reference is~\cite{1055638}, and~\cite{NIST-DSS}
contains some current standards).
A somewhat newer concept is \textit{quotable signature schemes}~\cite{kreu19}, which
are digital signature schemes with the additional property
that signatures are \textit{quotable} in the following sense.
The \emph{Signer} of a message $m$ generates a quotable signature $s$ for $m$
using a private key $\mathsf{sk}$.
Given a message~$m$ and the quotable signature $s$,
a \emph{Quoter}
(any third party)
can extract a second quotable signature $s'$
for a quote $q$ from $m$
without knowing $\mathsf{sk}$ or interacting with the original \emph{Signer}.
A quote can be any \textquote{allowable subsequence} of $m$.
We write $q\preceq m$ to indicate that $q$ is a quote from $m$.
This quotable signature $s'$ is still signed
with the private key $\mathsf{sk}$ of the \emph{Signer}
and hence authenticates the original \emph{Signer} as the author of the quote.
These signatures for quotes have the same required properties with respect to
verification and security as a standard digital signature, in addition to
allowing one to derive where content has been removed, relative to the quote.
A signature for a quote is again a quotable signature
with respect to sub-quotes of the quote,
and neither authenticating a quote nor sub-quoting require access to the original
message.

Quotable signatures can be used to
mitigate the effects of fake news and disinformation.
These are not new problems, and it is becoming increasingly apparent that they
are posing a threat for democracy and for society.
There is not one single reason for this, but one reason among many is a
fundamental change in how news is consumed: a transition is happening, where
explicit news products such as printed newspapers and evening news programs are
still consumed, but are increasingly giving way for shorter formats and snippets
of news on social media platforms~\cite{reutersNewsReport2019}.
However, people tend to be unable to recall from which news brand a story
originated when they were exposed to it on social
media~\cite{Newsbrandattribution18}.
This is problematic since the news media's image is an important heuristic when
people evaluate the quality of a news
story~\cite{doi:10.1080/1461670X.2013.856670}.
In addition, according to the Reuters Institute Digital News Report 2022~\cite{Reuters}, across markets, 54\% of those surveyed say they worry about identifying the difference between what is real and fake on the Internet when it comes to news, but people who say they mainly use social media as a source of news are more worried (61\%).

In recent years, a common approach to fighting back against fake news has been
flagging (potentially) fake news, using either manual or automatic detection
systems.
While this might be a natural approach, research has shown repeatedly that
flagging problematic content tends to have the opposite result, i.e., it
increases the negative effects of fake
news~\cite{2020ERCom...2h1003D,doi:10.1177/1529100612451018,10.1007/978-3-030-61841-4_16}.
This indicates that flagging problematic content is not sufficient and
alternative approaches need to be developed.

We present a method that complements flagging problematic content
with the goal
of mitigating the effect of fake news.
Our idea builds on the observation that \textit{which} news media published a
news article is an important heuristic people use to evaluate the quality of the
article~\cite{doi:10.1080/1461670X.2013.856670}.
However, since people get their news increasingly via social media, it
is becoming more likely that they are not aware of who published the
news they are consuming.
To address this, we propose using quotable signatures to allow people on social
media to find out and be certain of where the text they are reading originates
from, and to verify that any modifications to the text were all allowed.
Specifically for news, the proposed idea is that a news media publishing an
article also publishes a quotable signature for the article signed with their
private key.
When someone shares a quote from the article, they then also include the
signature for the quote that is derived from the initial signature (without
access to the private key), which we emphasize is signed with the same key.
Finally, when one reads the quote, the signature can be checked, and it can be
verified from where the quote originates.

The idea of mitigating the effects of fake news and misinformation, by using
digital signatures to verify the source of media content, is one that has been
addressed by others.
One example is C2PA~\cite{C2PA}, which involves many companies, including Adobe,
the BBC, Microsoft, and Twitter.
C2PA focuses on providing a history of a published item, i.e.,
which device was used to capture it, how it has been edited and by whom, etc.
Thus, quotable signatures could be of interest to their approach.

Another issue involving fake news
is that news articles are perceived as more credible
if they contain attributed quotes~\cite{doi:10.1177/107769909807500108}.
This is misused by fake news to appear more credible by providing attributions for
their
content~\ShortLong{\cite{kenya,fakeGermReport,fakeSchmeichel,fakeBBCScreenshot,fakeBr}}{\cite{kenya,fakeGermReport,fakeSchmeichel,fakeBBCScreenshot,fakeBr,fakeWaPo}},
but can in turn be used to automatically detect fake news by considering the
existence and quality of
attributions~\cite{ABUSALEM2021100369,DBLP:conf/semco/TraylorSGS19,doi:10.1177/0002764219878224}
(among other things).
Quotable signatures, in contrast, could be used to sign quotes to make a strong
and verifiable connection between the original source and the quote.
On the other hand, fake news would generally not be able to link their quotes to
reputable sources, thereby providing another heuristic helping users to distinguish
between authentic and fake content.

Without major changes to the system, it could be extended to further
use cases such as signing Facebook and Twitter posts, official
governmental rules and regulations, scientific publications, etc.  For
all of these instances, an important feature of our system that we
have not used explicitly so far is that signing also binds the Signer,
meaning that the signing party cannot later deny having signed the
signed document.



We provide an overview over related work in \cref{sec:priorArt}.
In \cref{sec:QuotSigsMain}, we give a more thorough introduction to and definition of quotable signatures, and we show how we can realize quotable signatures using Merkle trees~\cite{Merkle80,Merkle89}.
We define a notion of security for quotable signature schemes, and prove that the notion is satisfied by our construction.
Additionally, we prove a number of bounds on the size and computational costs of quotable signatures obtained using Merkle trees.
\ShortLong{}{Finishing off the construction of quotable signatures from Merkle trees, we describe algorithms for signing, quoting, and verifying in \cref{sec:impl}.}
We revisit the application of quotable signatures to counter fake news in more detail in \cref{sec:QSandFN} and we conclude the paper with an outlook to future work in \cref{sec:futureWork}.
\ShortLong{ In the full version of this paper we also describe concrete algorithms for our construction of quotable signatures from Merkle trees.}{}


\section{Related Work}\label{sec:priorArt}
Quotable signatures have been introduced in~\cite{kreu19}, which
suggests
constructing quotable signatures using Merkle trees and provides a rudimentary
complexity analysis. The authors also suggest using quotable signatures to mitigate the
effects of fake news. Compared to~\cite{kreu19}, we define a security model, and
prove that our construction is secure in this security model. Additionally, we
also provide proofs of our claims about the cost of using Merkle trees for
quotable signatures,\ShortLong{}{ provide concrete algorithms for quotable signatures from Merkle trees,}
and provide more in-depth considerations for why one could expect
this to be a good approach.


A concept closely related to quotable signature schemes is
\textit{redactable signature schemes} (RSSs). Simultaneously introduced
in~\cite{DBLP:conf/icisc/SteinfeldBZ01} (as \textit{Content Extraction
Signatures}) and~\cite{DBLP:conf/ctrsa/JohnsonMSW02}, RSSs essentially allow an
untrusted redactor to remove (\textquote{redact}) parts of a signed message,
without invalidating the signature. Often this requires modifying the signature,
but crucially, it is still signed with the original key, despite the redactor
not having access to the private key. Thus, quotable signatures share many
similarities with RSSs; if one considers a quotation as a redaction of all parts
of a text except for the quote, they are conceptually identical. Where quotable
signatures and RSSs differ is in the security they must provide. Both signature
schemes require a similar notion of unforgeability, but
an RSS must also guarantee that the redacted parts remain private. A standard
formulation is that an outsider not holding any private keys should
\textquote{not be able to derive any information about redacted parts of a
  message}, and even stronger requirements, such as transparency or
unlinkability, are not uncommon~\cite{DBLP:conf/IEEEares/BilzhausePS17}. Quotable
signatures have no such privacy requirements, allowing quotable signatures to be
faster. In fact, it is worth noting that there are scenarios where RSSs' notion
of privacy would be directly harmful to a quotable signature.
For instance, RSS would specifically make it
impossible to tell if a quote is contiguous or not, something that we consider
essential for a quotable signature scheme.
To see the value of dropping the privacy requirement, we observe that some RSSs with
$O(n)$ performance may have $O(n)$ expensive public key cryptography
operations~\cite{DBLP:conf/acns/BrzuskaBDFFKMOPPS10,DBLP:conf/acns/SamelinPBPM12},
whereas quotable signatures can be obtained with $O(n)$ (cheap) symmetric
cryptographic operations (hashing), and only one expensive public key operation.
There are approaches obtaining RSSs using only one expensive operation, but they
either require many more cheap operations than quotable signatures do, or they result in considerably larger signatures, for example~\cite{DBLP:conf/secrypt/HiroseK13}.
Early examples of RSSs had a weaker notion of privacy, but
still stronger than what we require. They require only hiding
of the redacted elements, not their location and number. Examples can be found
in~\cite{DBLP:conf/ctrsa/JohnsonMSW02,DBLP:conf/icisc/SteinfeldBZ01}. Their
approaches are similar to ours, also using Merkle trees, but we provide rigorous
proofs of the claimed performance, and our lack of privacy requirements
allows our scheme to be both more efficient and conceptually simpler. One
consideration that is very relevant for quotable signatures, \ShortLong{}{but seldom considered elsewhere, }is how a quote (redaction) being contiguous
will affect the complexity results. In a different
setting~\cite{DBLP:conf/dbsec/DevanbuGMS00} considers this question for Merkle
trees, but provides no rigorous proof.


Considering the motivating example again, approaches to mitigate the impact of
fake news, using either digital signatures or directly rating the source of the
content, have been proposed and tried before. One approach, serving as
inspiration for our approach, is~\cite{9912787}. They use digital
signatures to verify the authenticity of images and other forms of multimedia.
One drawback of their implementation is that it requires the media to be
bit-for-bit identical to the version that was signed. Hence, the image can for
instance not be compressed or resized, and thus their solution is not compatible
with many platforms,
e.g., Facebook compresses uploaded images,
and many news websites resize images for different screen sizes.
An example of directly rating the source of content, and flagging trustworthy
sources, can be found in \textquote{NewsGuard Ratings} (NG), which provides a
rating of trustworthiness for news sources. NG adds a flag that indicates if a
news source is generally trustworthy (green) or not (red) to websites and
outgoing links on websites. This approach has not been widely successful. For
example, the study in \cite{doi:10.1126/sciadv.abl3844} shows that NG's labels
have \textquote{limited average effects on news diet quality and fail to
  reduce misperceptions}. While this is somewhat related to our approach, there
are two major differences. (1) NG only flags content that directly links to the
source of the content with a URL.
In contrast, our digital signature can be attached to any text quote. Hence, NG
only adds additional information when it is already straightforward to figure
out from where the content originates. Our approach also provides this
information where there might otherwise be no clear context. (2) NG focuses on
providing a rating for how trustworthy a news source is.
This approach is similar to the
typical approach of telling people when something might be problematic, which
tends to have the opposite result. 
In contrast, we focus solely on providing and
authenticating the source of a quote.

Summing up, the contributions of this paper is as follows. (1) We rigorously
define the notion of security that quotable signature schemes must satisfy. (2)
We rigorously prove the security of and analyze the complexity of, a quotable
signature scheme constructed using Merkle trees. (3) This provides a scheme for
quotable signatures that is more efficient than using an RSS for the same
purpose. \ShortLong{}{(4) We provide concrete algorithms for quotable signatures using Merkle trees.}

\section{Quotable Signatures}\label{sec:QuotSigsMain}

To construct a quotable signature scheme, we follow
the approach suggested in~\cite{kreu19}
and use a combination of a classical digital
signature scheme~\cite{1055638}
and {Merkle trees}~\cite{Merkle80,Merkle89}.

Before getting into the construction,
we summarize the setting of quotable signatures.
In \cref{ssec:security-model}, we define
the security notion that quotable signature schemes should satisfy.
Then, in \cref{ssec:M-trees}, we introduce Merkle trees,
in \cref{ssec:get-q-sigs} we construct a
quotable signature scheme and show it is secure,
and finally we analyze the complexity of the scheme
in \cref{ssec:m-tree-perf}.

\paragraph{General setting for quotable signatures.}
A quotable signature scheme consists of four efficient algorithms, \textsf{QS =
  (KeyGen, Sign, Quo, Ver)}. These four algorithms are essentially the standard
three algorithms from a classical digital signature scheme
for key generation, signing, and verification,
with the added quoting algorithm $\mathsf{Quo}$.
To quote from a message, $\mathsf{Quo}$ allows
extracting a valid signature for the quote from the signature of the message
in such a way that it is still signed with the public key used to sign the
original message.
Additionally, it should be possible to derive from the signature of a quote
where tokens from the original message
have been removed
relative to the quote.

We refer to the involved parties
as the \textit{Signer}, the \textit{Quoter}, and the \textit{Verifier}.
We use $\lambda$ to denote the security parameter.
To summarize:
\begin{itemize}
\item $(\mathsf{sk,pk}) \leftarrow \mathsf{KeyGen}(1^{\lambda})$ takes as input the security
  parameter $1^\lambda$. It outputs a public key pair. This is typically done by the Signer once, offline as part of the initial setup.
\item $s \leftarrow \mathsf{Sig}_{\mathsf{sk}}(m)$ takes as input a secret key
  $\mathsf{sk}$ and a message $m$. It outputs a quotable signature for $m$. This is done by the Signer.
\item $s' \leftarrow \mathsf{Quo}(m,q,s)$ takes as input a message $m$, a quote $q$ from $m$, and a quotable signature $s$ for $m$.
  It outputs a quotable signature $s'$ for $q$, that is still signed with the secret key used to generate $s$.
  Verifying $s'$ does not require knowing~$m$.
  Note that $m$ and $s$ could have been obtained via an earlier quote operation.
  This is done by the Quoter.
\item $\top/\bot \leftarrow \mathsf{Ver}_{\mathsf{pk}}(q,s')$ takes as input a public key
  $\mathsf{pk}$, a quote (message) $q$, and a signature $s'$ for $q$. It outputs
  $\top$ if $s'$ is a valid signature for $q$ with respect to $\mathsf{pk}$, and
  $\bot$ otherwise. This is done by the
  Verifier.
\end{itemize}
%
\Cref{fig:general-setting} illustrates the typical interactions between the parties.

\begin{figure}[t!]
  \centering
  \resizebox{.7\columnwidth}{!}{\begin{tikzpicture}[>=stealth]
      \node [] (m) at (0,1.2) {$m$};
      \node [] (q) at (4,1.2) {$q$};
      \node [rectangle, align=center, text centered, draw=black, very thick,
      minimum height=1cm, minimum width=2cm] (Signer) at (0,0) {\textbf{Signer}\\$s \leftarrow \mathsf{Sig}_{\mathsf{sk}}(m)$};
      \node [rectangle, align=center, text centered, draw=black, very thick,
      minimum height=1cm, minimum width=2cm] at (4,0) (Quoter)
      {\textbf{Quoter}\\$s' \leftarrow \mathsf{Quo}(m,q,s)$};
      \node [rectangle, align=center, text centered, draw=black, very thick,
      minimum height=1cm, minimum width=2cm] (Verifier) at (8,0) {\textbf{Verifier}\\$\mathsf{Ver}_{\mathsf{pk}}(q,s')$};
      \node [] (a) at (8,-1.2) {Accept/Reject};
      \draw [->] (m) -- (Signer);
      \draw [->] (q) -- (Quoter);
      \draw [->] (Signer) -- node [above] {$m,s$} (Quoter);
      \draw [->] (Quoter) -- node [above] {$q,s'$} (Verifier);
      \draw [->] (Verifier) -- (a);
      \draw [->] ($(Quoter.east)+(0,-.4)$) -- ++(.2,0) -- ++(0,-.4) node (tmp) {}
                 --  node [below] {$\stackrel{m \leftarrow q}{s \leftarrow s'}$} ($(Quoter.south west |- tmp)+(-.25,0)$)
                 |- ($(Quoter.west)+(0,-.4)$);
    \end{tikzpicture}}
  \caption{The general setting for a quotable signature.}
  \label{fig:general-setting}
\end{figure}
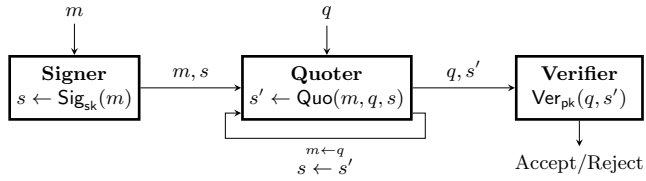

\subsection{Security Model}\label{ssec:security-model}
Taking inspiration from the RSS notion of unforgeability,
we define the security notion of quotable signatures schemes in
\cref{def:unforg}.
At its core, this is the standard notion of unforgeability for digital
signature schemes, with the additional requirement that the adversary's chosen
message cannot be a quote from any of the messages that the adversary sent to
the signing oracle.

\begin{definition}[Unforgeability]\label{def:unforg}
  Let $\mathsf{QS} = (\mathsf{KeyGen, Sign, Quo, Ver})$ be a quot\-able signature
  scheme.
  We say that $\mathsf{QS}$ is \emph{existentially unforgeable}, if for every
  probabilistic polynomial time adversary $\mathcal{A}$, the probability of the
  following experiment returning $1$ is negligible:\\[-2mm]

  \begin{pchstack}
    \pseudocode[mode=text]{
      $(\mathsf{pk,sk}) \leftarrow \mathsf{KeyGen}(1^{\lambda})$\\
      $(m^{*},s^{*}) \leftarrow \mathcal{A}^{\mathsf{Sign}_{\mathsf{sk}}(\cdot)}(\mathsf{pk})$\\
      \pccomment{denote the queries that $\mathcal{A}$ make to the signing
        oracle by $m_1,m_2,\ldots,m_Q$.}\\
      \pcif $(\mathsf{Ver}_{\mathsf{pk}}(m^{*},s^{*}) = \top) \wedge (\forall k \in
      \{1,2,\ldots,Q\} \colon m^{*} \not\preceq m_{k})$\\
      \t \t \pcreturn 1}
  \end{pchstack}
\end{definition}

\subsection{Merkle Trees}\label{ssec:M-trees}
A Merkle tree (also known as a \textit{hash tree}) allows one to efficiently
and securely verify that one or more \emph{tokens} are contained in a longer sequence
of tokens, without having to store the entire sequence~\cite{Merkle80,Merkle89}.
Examples of this could
be words forming a sentence, sentences forming an article, or data blocks making up a file.




Since our scheme will rely on hash functions, we assume that the tokens are
binary strings. Equivalently, one could assume an implicitly used, well defined
injective mapping from the token space to the space of binary strings. For data
blocks, the identity mapping suffices and for words one such mapping could be
the mapping of words to their UTF-8 representations.

The structure of a Merkle tree for a sequence of tokens is a binary tree, where
each leaf corresponds to a token from the sequence, with the leftmost leaf
corresponding to the first token, its sibling corresponding to the second token,
and so on. Each leaf is labeled with the hash of its token and each internal
node is labeled with the hash of the concatenation of the labels of its
children. Hence, the $i$'th internal node on the $j$'th level will be labeled as
\begin{align}
  u_{j,i} = H(u_{j+1,2i} \concat u_{j+1,2i+1}).
\end{align}

This way, one can show that any specific token is in the sequence by providing
the \textquote{missing} hashes needed to calculate the hashes on the path from
the leaf corresponding to the token to the root of the tree. 
Following established terminology, we call this the
\textit{verification path} for the token.\footnote{This use of
  \textquote{path} is slightly counter intuitive, since it refers to the
  hashes needed to calculate the hashes on the path from the leaf to the root,
  and hence not the nodes on this path but their siblings.}

\Cref{fig:Merkle-tree-jumps} shows the Merkle tree for a sequence
of words forming the sentence \textquote{The quick brown fox jumps over the
  dog}. The verification path for the word \textquote{jumps}
consisting of nodes $u_{3,5}$, $u_{2,3}$, and $u_{1,0}$
is highlighted in red.
Similarly, one can obtain the verification path for a subsequence of more than
just one token. In \cref{fig:Merkle-tree-jumps}, we also indicate
the verification path for the contiguous subsequence \textquote{the quick} in
blue. Note that the size of the verification path depends
not only on how many tokens are chosen, but also on where in the sequence they
are placed. In \cref{ssec:m-tree-perf}, we analyze how large the
verification path can become,
i.e., how many nodes need to be provided in the signature in the worst case.


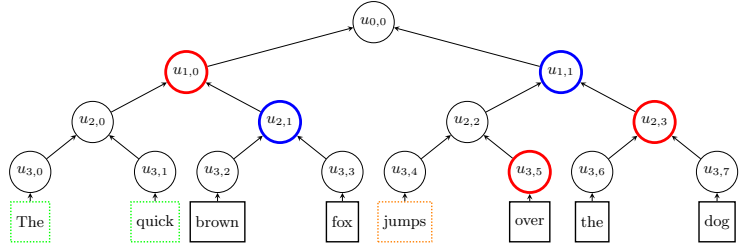
\begin{figure}[t!]
  \centering
  \resizebox{.8\columnwidth}{!}{%
    \begin{tikzpicture}[<-,>=stealth,level/.style={sibling distance = 7.5cm/#1, level distance = 1cm}]
      \node [treenode] (root) at (0,0) {$u_{0,0}$}
      child{ node [s-treenode] {$u_{1,0}$}
        child{ node [treenode] {$u_{2,0}$}
          child{ node [treenode] {$u_{3,0}$}
            child{ node [leaf, draw=green, densely dotted ] {The}
            }
          }
          child{ node [treenode] {$u_{3,1}$}
            child{ node [leaf, draw=green, densely dotted ] {quick}
            }
          }
        }
        child{ node [treenode, draw=blue, ultra thick] {$u_{2,1}$}
          child{ node [treenode] {$u_{3,2}$}
            child{ node [leaf] {brown}
            }
          }
          child{ node [treenode] {$u_{3,3}$}
            child{ node [leaf] {fox}
            }
          }
        }
      }
      child{ node [treenode, draw=blue, ultra thick] {$u_{1,1}$}
        child{ node [treenode] {$u_{2,2}$}
          child{ node [treenode] {$u_{3,4}$}
            child{ node [leaf,draw=orange, densely dotted ] {jumps}
            }
          }
          child{ node [s-treenode] {$u_{3,5}$}
            child{ node [leaf] {over}
            }
          }
        }
        child{ node [s-treenode] {$u_{2,3}$}
          child{ node [treenode] {$u_{3,6}$}
            child{ node [leaf] {the}
            }
          }
          child{ node [treenode] {$u_{3,7}$}
            child{ node [leaf] {dog}
            }
          }
        }
      };
    \end{tikzpicture}
  }
  \caption{An example of a Merkle tree where the tokens are words and the
    sequence is a sentence. The verification path for the token
    \textquote{jumps} is highlighted in red ($u_{1,0}, u_{2,3}, u_{3,5}$), and
    the verification path for the subsequence \textquote{The quick} is
    highlighted in blue ($u_{1,1},u_{2,1}$).}
  \label{fig:Merkle-tree-jumps}
\end{figure}

In these examples, we have chosen a sequence of tokens
where the length of the sequence, i.e., the number of tokens, is a power of two.
If the sequence length is not a power of two,
we require that the tree is \emph{heap-shaped},
i.e., all levels are filled, except for possibly the lowest level,
which is filled from the left up to some point, after which the lowest level
is empty.
%

\begin{samepage}
  \begin{remark}
    Observe that from the structure of the Merkle tree, one can see where in the
    sequence the quoted tokens are placed, and if they are sequential or
    discontinuous.
  \end{remark}
\end{samepage}

\subsection{A Quotable Signature Scheme}\label{ssec:get-q-sigs}
Using a Merkle tree, we can now
devise a scheme by which the Quoter can convince the Verifier that some quote is
contained in a larger text, if the Verifier is already in possession of the root
hash. The Quoter simply shares the verification path together with the quote,
and the Verifier verifies that this indeed leads to the original root hash. In
order to turn this into a quotable signature scheme, we include a classical
digital signature for the root hash, signed by the Signer, with the verification
path. Thus, letting \textsf{DS = (KeyGen$^{\mathsf{DS}}$, Sign$^{\mathsf{DS}}$,
  Ver$^{\mathsf{DS}}$)} be a classical digital signature scheme, our quotable
signature scheme can be described as follows:
\begin{itemize}
\item \textsf{KeyGen}: Identical to $\mathsf{KeyGen^{DS}}$.
\item \textsf{Sign}: Find the root hash of the Merkle tree and sign it with
  $\mathsf{Sign^{DS}}$.
\item \textsf{Quo}: Find the verification path of the quote. Together with the
  signature of the root hash, this forms the signature for the quote.
\item Find the root hash of the Merkle tree using the quote and its verification
  path. Use $\mathsf{Ver^{DS}}$ to verify the authenticity of the root hash.
\end{itemize}


\subsubsection{Proof of Security}\label{ssec:mt-qs-security}
We will show that the construction of the previous section is secure with
respect to the notion of security introduced in \cref{def:unforg}, when
instantiated with a secure hash function and a secure classical signature
scheme. Before doing so, we observe that currently, our scheme is trivially
vulnerable to a forgery attack, as follows. An adversary obtains a
quotable signature for a message from a signing oracle and then simply replaces
the last two tokens on the lowest level with a single token, which is the concatenation of the
tokens' hashes.
\ShortLong{For example, using this attack on the message used in \cref{fig:Merkle-tree-jumps}, would give the message \textquote{The quick brown fox jumps over $H(\text{the})\concat{}H(\text{dog})$}.}{We illustrate this in \cref{fig:2nd-preimage-attack}, where we
have created a second preimage of the message used in
\cref{fig:Merkle-tree-jumps}.} However, there is an easy fix to this
vulnerability. Noting that the problem is that an adversary can claim
that an internal node is a leaf, we can prevent this by applying domain
separation in the form of adding one value to the leaves before hashing, and
another value to the internal nodes before hashing. Taking inspiration from
RFC 6962~\cite{rfc6962}, the Merkle trees are modified by prepending $00$ to
the leaves before hashing and $01$ to the internal nodes before hashing.
From now on, we implicitly assume that this is done.

\ShortLong{}{\begin{figure}[!t]
  \centering
  \resizebox{.8\columnwidth}{!}{%
    \begin{tikzpicture}[<-,>=stealth,level/.style={sibling distance = 7.5cm/#1, level distance = 1cm}]
      \node [treenode] (root) at (0,0) {$u_{0,0}$}
      child{ node [treenode] {$u_{1,0}$}
        child{ node [treenode] {$u_{2,0}$}
          child{ node [treenode] {$u_{3,0}$}
            child{ node [leaf] {The}
            }
          }
          child{ node [treenode] {$u_{3,1}$}
            child{ node [leaf] {quick}
            }
          }
        }
        child{ node [treenode] {$u_{2,1}$}
          child{ node [treenode]{$u_{3,2}$}
            child{ node [leaf] {brown}
            }
          }
          child{ node [treenode] {$u_{3,3}$}
            child{ node [leaf] {fox}
            }
          }
        }
      }
      child{ node [treenode] {$u_{1,1}$}
        child{ node [treenode] {$u_{2,2}$}
          child{ node [treenode] {$u_{3,4}$}
            child{ node [leaf] {jumps}
            }
          }
          child{ node [treenode] {$u_{3,5}$}
            child{ node [leaf] {over}
            }
          }
        }
        child{ node [treenode] {$u_{2,3}$}
          child{ node [leaf] {$H(\text{the})\concat{}H(\text{dog})$}
          }
        }
      };
    \end{tikzpicture}
  }
  \caption{A Merkle tree for the sequence \textquote{The quick brown fox jumps
      over $H(\text{the})\concat{}H(\text{dog})$}, which is a second preimage to the
    Merkle tree for the sequence \textquote{The quick brown fox jumps over the
      dog}.}
  \label{fig:2nd-preimage-attack}
\end{figure}
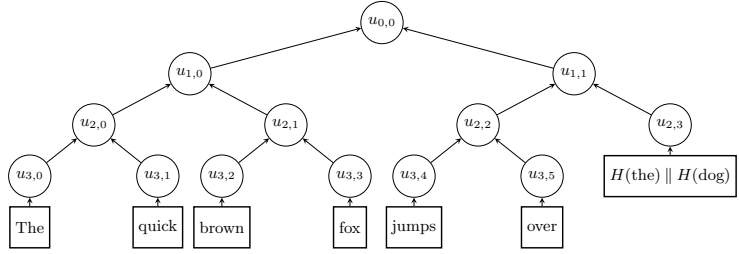}

We can now argue that the construction is secure.

\begin{theorem}\label{thr:unforg}
  Under the assumption that
  \begin{itemize}
  \item $H$ comes from a family of cryptographic hash function,
  \item \textsf{DS = (KeyGen$^{\mathsf{DS}}$, Sign$^{\mathsf{DS}}$,
      Ver$^{\mathsf{DS}}$)} is an existentially unforgeable classical signature scheme,
  \end{itemize}
  $\mathsf{QS} = (\mathsf{KeyGen, Sign, Quo, Ver})$ constructed as described above,
  is an existentially unforgeable quotable signature scheme.
\end{theorem}

We have to show that no probabilistic polynomial time adversary can win the unforgeability experiment in \cref{def:unforg} with non-negligible probability.

\begin{proof}  
  Assume that $\mathcal{A}$ is a probabilistic polynomial time adversary against the unforgeability of
  \textsf{QS}. We show that the probability of $\mathcal{A}$ being successful is
  negligible.
  Let $(m^{*},s^{*})$ be the output of $\mathcal{A}$, where
  $s^{*} = (\mathsf{Sig^{DS}_{sk}}(u_{0,0}^{*}), \{u_{i,j}^*\})$, i.e., the
  classical digital signature of the root hash and a (possibly empty)
  verification path.

  Consider first the case where the root hash $u_{0,0}^{*}$ of $m^{*}$ (found
  using $\{u_{i,j}^*\}$) is different from the root hashes of the queries
  $\mathcal{A}$ made to the signing oracle. In this case,
  $(u_{0,0}^{*}, \mathsf{Sig^{DS}_{sk}}(u_{0,0}^{*}))$ is a forgery against
  $\mathsf{DS}$, and since \textsf{DS} is assumed to be existentially
  unforgeable, this can only happen with negligible probability. Denote this
  probability as $\epsilon_{\mathsf{DS}}$.

  If this is not the case, there must be an $m_k$, such that the root hash of
  $m_k$ is $u_{0,0} = u_{0,0}^{*}$, but $m^{*} \not\preceq m_k$. Denote by
  $T^{*}$ the tree for $m^{*}$ (constructed using the verification path, if one
  is included) and by $T$ the tree for $m_k$.

  Consider first the case where all leaves, corresponding to tokens, in $T^{*}$
  are at a location in the tree, where there is also a leaf, corresponding to a
  token, in $T$. Since $m^{*} \not\preceq m_k$ there must be tokens
  $a^{*},a$ such that $a^{*} \in m^{*}$ and $a \in m_k$
  are at the same positions in their respective trees, and
  $a^{*} \ne a$. Observe that if
  $H(\mathtt{00}\concat a^{*}) = H(\mathtt{00}\concat a)$, we have
  found a collision to $H$. If
  $H(\mathtt{00}\concat a^{*}) \ne H(\mathtt{00}\concat a)$, let the
  nodes on the path between the leaf corresponding to $a^{*}$ and the root
  of $T^{*}$ be denoted by
  $u_{i,j_i}^{*},u_{i-1,j_{i-1}}^{*},\ldots,u_{1,j_1}^{*},u_{o,o}^{*}$ and the
  nodes on the path between the leaf corresponding to $a$ and the root of
  $T$ by $u_{i,j_i},u_{i-1,j_{i-1}},\ldots,u_{1,j_1},u_{o,o}$. Since
  $u_{i,j_i}^{*} \ne u_{i,j_i}$ and $u_{o,o}^{*} = u_{o,o}$, there
  exists a $0 \le \ell < j$ such that $u_{\ell,j_\ell}^{*} = u_{\ell,j_\ell}$ and
  $u_{\ell+1,j_{\ell+1}}^{*} \ne u_{\ell,j_{\ell+1}}$. Thus, $u_{\ell+1,j_{\ell+1}}^{*} $ and
  $ u_{\ell,j_{\ell+1}}$ (together with their siblings and $\mathtt{01}$) form a
  collision.

  Consider now the case where there is a leaf, corresponding to a token, in
  $T^{*}$ that is not at a location in the tree, where there is a leaf,
  corresponding to a token, in $T$. In this case there must be nodes
  $u_{i,j}^{*} \in T^{*}$ and $u_{i,j} \in ^T$ at the same position in their
  respective trees such that one of them is internal and the other
  corresponds to a token. If $u_{i,j}^{*}$ and $u_{i,j}$ do not have the same
  label, we can apply the method from the precious paragraph to find a
  collision. If they have the same label, we must have two nodes
  $u_{i+1,2j}, u_{i+1,2j+1}$ in $T$ or $T^{*}$, and a token $a$ in $m^{*}$ or
  $m_i$ such that
  $H(\mathtt{01}\concat{}u_{i,j}\concat{}u_{i,j+1}) = H(\mathtt{00}\concat{}a)$,
  and we have found a collision.

  We observe that in all cases, we have found a collision for $H$. Since $H$ is
  assumed to be secure, and hence collision resistant, this can happen only
  negligible probability. Denote this probability as $\epsilon_{H}$.


  Hence,
  $\mathcal{A}$'s advantage
  of at most $\epsilon_{\mathsf{DS}} + \epsilon_{H}$
  is negligible.
\end{proof}

\subsection{Performance}\label{ssec:m-tree-perf}
\cref{tab:theo-perf} shows the cost of our quotable signature scheme
for each of the three parties.
This is measured in terms of computation
due to the number of required hash operations
and classical signature operations
as well as in terms of the size of the generated signature
due to the required hash values and classical signatures,
presumably the dominant operations.
In all cases, we assume that the message $m$ has length $n$, i.e.,
$m$ consists of $n$ tokens. For the Quoter and the Verifier, we additionally
assume that the quote has length $t \le n$. 

\begin{table}[t]
  \centering
  \caption{Theoretical bounds on the performance of our version of a quotable
    signature. For the Quoter, we consider both if we allow quoting arbitrary
    tokens from the sequence, and when we require that the quoted tokens must be
    consecutive.}
  \label{tab:theo-perf}
  \small
  \setlength\tabcolsep{4pt}
  \begin{tabular}{l|cc}
    & \textbf{Computation Time} & \textbf{Signature Size} \\ \hline
    \textbf{The Signer} & $2n-1$ hashes and & 1 classical signature\\
    & 1 classical signature & \\ \hline
    \textbf{The Quoter} & \\
    Arbitrary & $2n-1$ hashes & 1 classical signature, at most \\
    && $t(\lceil \log{n} \rceil - \lceil \log{t}\rceil - 1)$ \\
    && $+ 2^{\lceil \log{t} \rceil}$ hashes \\
    Consecutive & $2n-1$ hashes & 1 classical signature, at most \\
    && $2 \lceil \log{n} \rceil - 2$ hashes \\ \hline
    \textbf{The Verifier} & 1 classical verification & ---\\
    & and up to $2n-1$ hashes &
  \end{tabular}
\end{table}

To put the results into context, running the command \texttt{openssl speed} on a modern
laptop shows that it is capable of computing
hundreds of thousands or even millions of hashes every second (depending on
the size of the data being hashed and the hash algorithm being used).
Additionally, a classical digital signature only takes a fraction of a second
create or verify. Thus, it is nearly instantaneous to generate/quote/verify a
quotable signature, even for sequences and quotes that are thousands of
tokens long.

The cost for the Signer, the Quoter, and the Verifier
is derived as follows.

\subsubsection{The Signer}
\ShortLong{}{Computing the cost for the Signer is straightforward. }To generate the Merkle
tree, the Signer needs to compute $2n-1$ hashes. To create the quotable digital
signature for $m$, she creates a classical digital signature for the root hash.
This classical digital signature is the Signer's signature
for her message $m$.

\subsubsection{The Quoter}
The Quoter also has to generate the entire Merkle tree, from which he can
extract the verification path for the quote he wishes to make. However, the
size of the verification path (and hence the signature for the quote) depends
on the size of the quote, and where in the text the quote is located. The most
simple case is when just one token is quoted, in which case the size of the
verification path is at most $\lceil \log{n} \rceil$, which, together with the
classical signature for the root hash, forms the signature for the quote.
Similarly, as shown in the following, the worst case can be
obtained by quoting every second token, in which case the Quoter would need
$\left\lceil \frac{n}{2} \right\rceil$ hashes on the verification
path.\footnote{Of course, algorithms can be adapted to include the entire text instead in such (rare) cases where that might require less space.}

In \cref{prop:size-of-signature-arbitrary}
we quantify the worst-case size of the verification path (and
hence the signature) for the quote in terms 
of message and quote lengths.
In \cref{prop:size-of-sig-contiguous}, we consider the special case
where we require that the quote be contiguous.

\begin{proposition}\label{prop:size-of-signature-arbitrary}
  For a message $m$ of size $n$ tokens and a quote of size $t$ tokens, the
  worst-case size of the verification path of the quote is at most
  \begin{align}
    t(\lceil \log{n} \rceil - \lceil \log{t}\rceil -1) + 2^{\lceil\log t\rceil}.
  \end{align}
\end{proposition}

\begin{proof}
  In \cref{lem:quote-when-n-is-pow2}, we consider the case where $n$ is a
  power of two. In this case, we identify a worst-case set of $t$ leaves of the
  Merkle tree on $n$ tokens.
  In \cref{lem:quote-when-n-is-not-pow2}, we establish that it is
  sufficient to consider $n$ a power of two.


  To argue about the size of the signature, we consider what we call the
  \textit{forest of independent trees} for a quote. To find the forest of
  independent trees for a quote, we do the following. For each token in the
  quote, consider the path between the node corresponding to that quote and
  the root (the root-token path). Define the \textit{independent tree
    corresponding to that token} to be the subtree rooted in the highest node
  on the root-token path, which is not on the root-token path for any other
  token in the quote.
  The forest of
  independent trees for the quote is now the collection of the independent
  trees of all the tokens in the quote. In \cref{fig:cut-level-merkle},
  we consider a message of size $n = 8$ and a quote of size $t =3$, quoting
  the first, third, and fifth token. The red line indicates a separation
  between the independent trees and the nodes that are on multiple root-token
  paths. The forest of independent trees consists of the trees rooted in
  $u_{2,0}, u_{2,1}$, and $u_{1,1}$.

  \begin{figure}[!t]
    \centering
    \resizebox{.8\columnwidth}{!}{%
      \begin{tikzpicture}[<-,>=stealth,level/.style={sibling distance = 7.5cm/#1, level distance = 1cm}]
        \node [treenode] (root) at (0,0) {$u_{0,0}$}
        child{ node [treenode] {$u_{1,0}$}
          child{ node [treenode] {$u_{2,0}$}
            child{ node [treenode] {$u_{3,0}$}
              child{ node [leaf, draw=orange] {\quad}
              }
            }
            child{ node [treenode] {$u_{3,1}$}
              child{ node [leaf] {\quad}
              }
            }
          }
          child{ node [treenode] {$u_{2,1}$}
            child{ node [treenode]{$u_{3,2}$}
              child{ node [leaf, draw=orange] {\quad}
              }
            }
            child{ node [treenode] {$u_{3,3}$}
              child{ node [leaf] {\quad}
              }
            }
          }
        }
        child{ node [treenode] {$u_{1,1}$}
          child{ node [treenode] {$u_{2,2}$}
            child{ node [treenode] {$u_{3,4}$}
              child{ node [leaf, draw=orange] {\quad}
              }
            }
            child{ node [treenode] {$u_{3,5}$}
              child{ node [leaf] {\quad}
              }
            }
          }
          child{ node [treenode] {$u_{2,3}$}
            child{ node [treenode] {$u_{3,6}$}
              child{ node [leaf] {\quad}
              }
            }
            child{ node [treenode] {$u_{3,7}$}
              child{ node [leaf] {\quad}
              }
            }
          }
        };
        \draw[-,dotted,thick,draw=red] (-7,-1.5) -- (0,-1.5) -- (0,-0.5) -- (7,-0.5);
      \end{tikzpicture}
    }
    \caption{A Merkle tree for a sequence of size $n = 8$ and a quote of size
      $t = 3$.}
    \label{fig:cut-level-merkle}
  \end{figure}
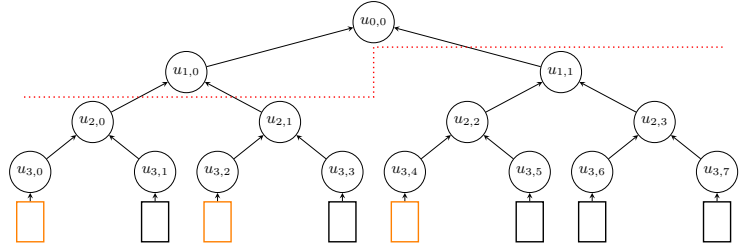

  \begin{lemma}\label{lem:pow2-max-diff-in-height}
    If $n$ is a power of two, the heights of the trees in the independent
    forest for a quote that maximizes the size of the signature can differ by at
    most $1$.
  \end{lemma}

  \begin{proof}
    Assume towards a contradiction that $Q$ is a quote that maximizes the size
    of the signature for $Q$ such that the difference between the heights of the smallest and
    largest trees in the forest of independent trees for $Q$ is at least $2$.
    Let $A$ be the root of a tree of minimal height in the forest of
    independent trees, and let $B$ be its sibling. Note that $B$ is also the
    root of a tree in the forest of independent trees (otherwise the tree
    rooted at $A$ would not be of minimal height). Additionally, let $C$ be
    the root of a tree of maximal height in the forest of independent trees.
    We illustrate this in \cref{fig:small-height-diff}.

    \begin{figure}[!t]
      \centering
      \resizebox{.8\columnwidth}{!}{%
        \begin{tikzpicture}[<-,>=stealth,level/.style={sibling distance = 3cm/#1, level distance = 1cm}]
          \node [treenode] (root) at (0,0) {$C$}
          child{ node [treenode] {$L$}
            child{ node [treenode] {$D_1$}
            }
            child{ node [treenode] {$D_2$}
            }
          }
          child{ node [treenode] {$R$}
            child{ node [treenode] {$D_3$}
            }
            child{ node [treenode] {$D_{4}$}
            }
          };
          \draw[-,dotted,thick,draw=red] (-1,.5) -- (1,.5);
          \draw[->,dotted] (root) -- (0,1);

          \node [treenode] (root2) at (-7,-1.5) {$P$}
          child{ node [treenode] {$A$}
          }
          child{ node [treenode] {$B$}
          };
          \draw[-,dotted,thick,draw=red] (-9,-2) -- (-5,-2);
          \draw[->,dotted] (root2) -- (-7,-.5);
        \end{tikzpicture}
      }
      \caption{Note that there might be trees rooted at $A, B, D_1, D_2, D_3,$
        and $D_4$, which we have omitted drawing, but by our assumption, the
        trees rooted at $D_1,D_2,D_3,$ and $D_4$ must be at least as high as the
        ones rooted at $A$ and $B$.}
      \label{fig:small-height-diff}
    \end{figure}

    Observe that we can now create a quote $Q'$ requiring more hashes than
    $Q$, by changing $Q$ in the following ways:
    \begin{itemize}
    \item Instead of quoting one token from the tree rooted at $A$ and one token
      from the tree rooted at $B$, $Q'$ quotes only one token from the tree
      rooted at $P$.
    \item Instead of quoting just one token from the tree rooted at $C$, $Q'$
      quotes one token from the tree rooted at $L$ and one token from the tree
      rooted at $R$. 
    \end{itemize}
    It is clear that $Q$ and $Q'$ quote equally many tokens and that the
    forest of independent trees for $Q'$ is only changed from the forest for
    $Q$ in the trees that involves $A,B,$ and $C$. The new situation is
    illustrated in \cref{fig:small-height-diff2}.

    \begin{figure}[!t]
      \centering
      \resizebox{.8\columnwidth}{!}{%
        \begin{tikzpicture}[<-,>=stealth,level/.style={sibling distance = 3cm/#1, level distance = 1cm}]
          \node [treenode] (root) at (0,0) {$C$}
          child{ node [treenode] {$L$}
            child{ node [treenode] {$D_1$}
            }
            child{ node [treenode] {$D_2$}
            }
          }
          child{ node [treenode] {$R$}
            child{ node [treenode] {$D_3$}
            }
            child{ node [treenode] {$D_{4}$}
            }
          };
          \draw[-,dotted,thick,draw=red] (-2,-0.5) -- (2,-0.5);
          \draw[->,dotted] (root) -- (0,1);

          \node [treenode] (root2) at (-7,-1.5) {$P$}
          child{ node [treenode] {$A$}
          }
          child{ node [treenode] {$B$}
          };
          \draw[-,dotted,thick,draw=red] (-8,-1) -- (-6,-1);
          \draw[->,dotted] (root2) -- (-7,-0.5);
        \end{tikzpicture}
      }
      \caption{Note that there might be trees rooted at $A, B, D_1, D_2, D_3,$
        and $D_4$, which we have omitted drawing, but by our assumption, the
        trees rooted at $D_1,D_2,D_3,$ and $D_4$ must be at least as high as the
        ones rooted at $A$ and $B$.}
      \label{fig:small-height-diff2}
    \end{figure}

    If each of the trees rooted at $A$ and $B$ contributed with $k$ hashes to
    $Q$, then the tree rooted at $C$ contributed with $k'+2$ hashes, where
    $k' \ge k$. In total, $A,B,$ and $C$ contributed $2k + k' + 2$ hashes.
    However, in $Q'$ we see that the tree rooted at $P$
    contributes $k+1$ hashes, and each of the trees rooted at $L$ and $R$
    contributes $k' + 1$ hashes, for a total of $k + 2k' + 3$ hashes. But since
    $k' \ge k$, we have that \ShortLong{$k + 2k' + 3 \ge 2k + k' + 3 > 2k + k' + 2$,}{ 
    \begin{align}
      k + 2k' + 3 \ge 2k + k' + 3 > 2k + k' + 2,
    \end{align}}
    contradicting that $Q$ maximizes the size of the signature.
  \end{proof}

  \begin{lemma}\label{lem:quote-when-n-is-pow2}
    When $n$ is a power of two, we can assume that the quote generating the
    largest signature has the properties that
    \begin{myclaim}
    \item\label{lem:quote-when-n-is-pow2-pt1} the heights of the trees in the independent forest for the quote
      differ by at most $1$,
    \item\label{lem:quote-when-n-is-pow2-pt2} for each tree in the forest of independent trees, the left-most leaf
      corresponds to the token that is quoted, and
    \item\label{lem:quote-when-n-is-pow2-pt3} the trees in the forest of independent trees are arranged with the
      smallest trees first.
    \end{myclaim}
  \end{lemma}

  \begin{proof}
    \cref{lem:quote-when-n-is-pow2-pt1} follows immediately from
    \cref{lem:pow2-max-diff-in-height}. Further,
    \cref{lem:quote-when-n-is-pow2-pt2} follows from observing that we can
    bring any tree to this form simply by swapping the children of some of the
    nodes on the path to the leaf corresponding to a quoted token (hereby
    changing which token is quoted, but not how many are quoted), and that these
    swaps do not affect the size of the signature. Finally,
    \cref{lem:quote-when-n-is-pow2-pt3} follows from observing that if two
    nodes are on the same level of the Merkle tree, and the labels of both are
    known, then we can \textquote{swap} the subtrees that they are roots of without
    affecting the size of the signature. By \textquote{swapping}, we mean that if the
    $i'$th leaf in the first node's subtree corresponds to a quote before the
    swap, then the $i$'th leaf in the second node's subtree corresponds to a
    quote after the swap, and so on. To see that this does not affect the size of the
    quote, note that outside of the two subtrees, nothing has changed; the hash
    of both nodes is still known. Additionally, from the first subtree we now
    get as many hashes as we got from the second subtree before the swap, and
    vice versa.
  \end{proof}

  \cref{lem:quote-when-n-is-pow2} implies that for any $n$ a power of two
  and $t\le n$, we need only consider one choice of which tokens are quoted.
  For example, \cref{fig:cut-level-merkle} shows the only quote of size
  $t = 3$ in a tree of size $n = 8$ that we need to consider.

  \begin{lemma}\label{lem:quote-when-n-is-not-pow2}
    For any message $m$ of length $n$ and quote $Q$ of length $t$, there is a
    quote $Q'$ of length $t$ from a message $m'$ of length
    $2^{\lceil \log{n} \rceil}$ such that the signature for $Q'$ is no smaller than the signature for $Q$.
  \end{lemma}

  \begin{proof}
    For fixed $m$ and $Q$, we create $m'$ by adding tokens to $m$ until
    $|m'| = 2^{\lceil \log{n} \rceil}$. We now create $Q'$ from $Q$ by going over each
    quote $q$ in $Q$.
    \begin{enumerate}
    \item If the leaf corresponding to $q$ in the Merkle tree for $m$ is on the
      deepest level, we quote the same token in $m'$.
    \item If the leaf corresponding to $q$ in the Merkle tree for $m$ is not on
      the deepest level, there is an internal node in the Merkle tree for $m'$
      at the location of the leaf in $m$. We quote the token corresponding to
      its left child, which is a leaf.
    \end{enumerate}
    Clearly, the tokens in $Q'$ from case
    1 contribute with the same number of hashes to the signature for $Q'$ as
    the corresponding ones did to the signature for $Q$, and the tokens from
    case 2 contribute with exactly one more hash.
    Hence, the
    signature for $Q'$ is at least as large as the signature for $Q$.
  \end{proof}

  We are now ready to derive the claim in
  \cref{prop:size-of-signature-arbitrary}. For any message $m$ and quote $Q$ we
  can assume that $|m| = n$ is a power of two, i.e.,
  $n = 2^{\lceil \log{n} \rceil}$ (otherwise
  \cref{lem:quote-when-n-is-not-pow2} allows us to instead consider an
  $m'$ that is a power of two), and that $Q$ has size $|Q| = t$ and exactly
  the structure described in \cref{lem:quote-when-n-is-pow2}.

  There are $t$ trees in the forest of independent trees for the quote, and
  all the way up to (but not including) their roots, each of these trees
  provides one hash per level. The roots of the trees in the forest are on the
  deepest level with less than $t$ nodes and the first level with more than
  $t$ nodes (if $t$ is a power of two, all roots are instead on the level with
  exactly $t$ nodes). Hence, all levels that are at depth more than
  $\lceil \log{t} \rceil$ contributes with $1$ hash per tree, for a total of
  \begin{math}
    t (\lceil \log{n} \rceil - \lceil \log{t} \rceil)
  \end{math}
  hashes. Additionally, we need to count how many hashes we get from the level
  at depth $\lceil \log{t} \rceil$. On this level, every node is either a root of an
  independent tree or a child of a root of an independent tree. In the first
  case, the hash of the node is calculable from information from lower
  levels. In the second case, for every pair of siblings, one of the nodes'
  hash is calculable from information from lower levels (the one on a
  root-token path for a token corresponding to a quoted token) and the other
  nodes' hash must be provided by the signature. Since there are $2^{\lceil \log{t}
    \rceil}$ nodes on this level, and $t$ independent trees, the
  signature must provide $2^{\lceil \log{t} \rceil} - t$ hashes on this level.

  In total, this shows that an upper bound on the number of hashes provided by
  the signature for a quote of $t$ tokens from an $n$ tokens sequence is
  \begin{align}
    &t(\lceil \log{n} \rceil - \lceil \log{t} \rceil) + 2^{\lceil \log{t} \rceil} - t \\
    =& t(\lceil \log{n} \rceil - \lceil
       \log{t} \rceil - 1) + 2^{\lceil \log{t} \rceil},
  \end{align}
  which finishes the proof of
  \cref{prop:size-of-signature-arbitrary}.
\end{proof}

\begin{corollary}
  For a message of size $n$ tokens and any quote, the
  worst-case size of the verification path of the quote is $\left\lceil \frac{n}{2}\right\rceil$.
\end{corollary}
  
Another easy corollary to the proof of \cref{prop:size-of-signature-arbitrary}---and
\cref{lem:quote-when-n-is-not-pow2} in particular---we can bound the
error when $n$ is not a power of two (when $n$ \textit{is} a power of two, the
bound is, of course, exact).

\begin{corollary}
  When $n$ is not a power of two, the bound of
  \cref{prop:size-of-signature-arbitrary} overcounts by at most $t$
  hashes.
\end{corollary}

\begin{proof}
  At each level of the Merkle tree, the signature needs to provide at most one
  hash for each quoted token. In the construction used in the proof of
  \cref{prop:size-of-signature-arbitrary} when $n$ is not a power
  of two, no levels are added to the Merkle tree, and hence the signature
  becomes no more than $t$ hashes larger.
\end{proof}

\begin{proposition}\label{prop:size-of-sig-contiguous}
  For a message of size $n > 2$ tokens and a contiguous quote of $t$
  tokens, the worst-case size of the verification path of the quote is
  \begin{math}
    2 \lceil \log{n} \rceil - 2
  \end{math}
  hashes.
\end{proposition}

\begin{proof}
  We prove this proposition by induction on the height of
  the Merkle tree.

  As the base case, we consider trees of height $2$. Either picking just one
  token or picking one token among the first two tokens and one token among the last
  one or two tokens, gives a verification path of worst-case size $2 \cdot 2 - 2 = 2$.

  Assume now that in a tree of height $k$, the largest possible size of the
  verification path for a contiguous quote is $2k - 2$. As our inductive step,
  we show that if the height of the Merkle tree of a message is $k+1$, then
  the largest possible size of the verification path for a contiguous quote
  from the message is $2(k+1) - 2$. For any contiguous quote $Q$, we consider
  two cases: (1) $Q$ is either contained in the first $2^k$ tokens or
  contains none of the first $2^k$ tokens, and (2) $Q$ contains both the
  $2^k$'th and the $(2^k+1)$'st token.

  \textbf{Case 1:}
    If $Q$ corresponds to leaves that are completely
    contained in one of the subtrees of the root, it follows from the
    induction hypothesis that the verification path consists of at most
    $2k - 2$ hashes from that subtree. The verification path
    contains only one additional hash, that of the root of the other subtree.
    Thus, the total number of hashes is at most $2k - 2 + 1 < 2(k+1) - 2$.

  \textbf{Case 2:} We make a few observations.
    Considering a level of the Merkle tree from left to right,
    the nodes with hashes that the Verifier calculates are consecutive.
    In \cref{fig:size-of-sig-contiguous}, we have illustrated this by
    highlighting in green all the nodes with labels that the Verifier
    calculates.

    Additionally, observe that for any level of depth $j\ge 2$, the
    only nodes of depth $j-1$ with a label that the Verifier has to
    calculate and that, at the same time, (potentially) has a child
    outside the consecutive sequence of nodes that the Verifier
    calculated the labels for at depth $j$, are the parents of the
    leftmost and rightmost nodes in that consecutive sequence at
    depth~$j$. All the nodes that might be characterized like this
    are on the two paths of black arrows in \cref{fig:size-of-sig-contiguous}.
    Hence, it follows that on each level, the verification
    path needs to provide at most $2$ hashes. Clearly, the root's
    label will not need to be provided by the verification path, and
    the root's children will also not need to have their labels
    provided since the quote contains a token from each child's
    subtree. Finally, observing that there are a total of $k+2$ levels
    in a tree of height $k+1$, allows us to conclude that the
    verification path needs to provide at most
      $2 \cdot (k + 2 - 2) = 2 \cdot (k+1) - 2$
    hashes, completing the case and the proof.
  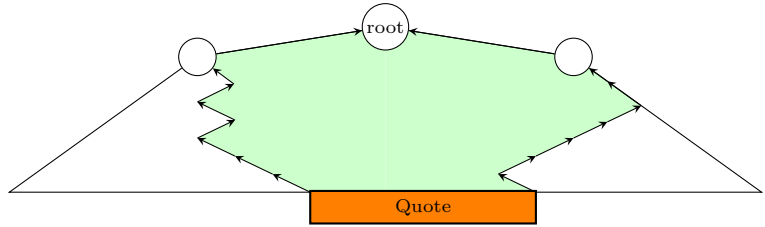
\begin{figure}[!t]
    \centering
      \begin{tikzpicture}[<-,>=stealth,xscale=.5,yscale=.2]
        \path[fill=green!20] (4,-10) -- (3,-8.8) -- (4,-7.6) -- (5,-6.4) -- (5.9,-5.325) --
        (6.8,-4.25) -- (5,-1) -- (0,1) -- (0,-10) -- (4,-10);
        \path[-,fill=green!20] (-2,-10) -- (-3,-8.8) -- (-4,-7.6) -- (-5,-6.4) -- (-4,-5.2) -- (-5,-4)
        -- (-4.03,-2.8) -- (-5,-1) -- (0,1) -- (0,-10) -- (-2,-10);
        \node[treenode,fill=white,text width=.6cm,thin] (root) at (0,1) {\scriptsize root};
        \node[treenode,fill=white,text width=.5cm,thin] (L) at (-5,-1) { };
        \node[treenode,fill=white,text width=.5cm,thin] (R) at (5,-1) { };
        \draw[-] (L) -- (-10,-10) -- (10,-10) -- (R) -- (root) -- (L);
        \draw (root) -- (L);
        \draw (root) -- (R);
        \node[leaf,fill=orange,thick, minimum height=.3cm,minimum width=3cm]
        (Q) at (1,-11) {\scriptsize Quote};
        \draw[->] (-2,-10) -- (-3,-8.8);
        \draw[->] (-3,-8.8) -- (-4,-7.6);
        \draw[->] (-4,-7.6) -- (-5,-6.4);
        \draw[->] (-5,-6.4) -- (-4,-5.2);
        \draw[->] (-4,-5.2) -- (-5,-4);
        \draw[->] (-5,-4) -- (-4.03,-2.8);
        \draw[->] (-4.03,-2.8) -- (L);
        \draw[->] (4,-10) -- (3,-8.8);
        \draw[->] (3,-8.8) -- (4,-7.6);
        \draw[->] (4,-7.6) -- (5,-6.4);
        \draw[->] (5,-6.4) -- (5.9,-5.325);
        \draw[->] (5.9,-5.325) -- (6.8,-4.25);
        \draw[->] (6.8,-4.25) -- (5.9,-2.625);
        \draw[->] (5.9,-2.625) -- (R);
      \end{tikzpicture}
    \caption{Merkle tree with a contiguous quote divided between the left- and
      right subtree. The labels of all the nodes in the green area are
      calculated from the labels of their children and do not need to be
      part of the signature.}
    \label{fig:size-of-sig-contiguous}
  \end{figure}
\end{proof}

\subsubsection{The Verifier}
The Verifier has to verify one classical digital signature and to reconstruct
the Merkle tree using the quote together with the verification path. Once
again, the cost of this depends on where in the message the quote is located,
with the number of hashes generally going towards $2n-1$ as the quote gets
closer to being the full message. For example, if all but one token has been
quoted, the Verifier needs to compute $2n-2$ hashes, and if only one token has
been quoted, the Verifier only needs to compute $\lceil \log{n} \rceil + 1$ hashes.


\ShortLong{}{\section{Design}\label{sec:impl}
This section considers some more practical aspects of quotable signatures from
Merkle trees. In \cref{ssec:design-of-ss}, we give an overview of
how the parties' algorithms work. Then, in \cref{ssec:choices}, we discuss
some application-specific choices one would have to make when implementing
or using quotable signatures.

\subsection{Algorithms}\label{ssec:design-of-ss}

Before going into the details of the algorithms,
we consider the intuitive approach one would take
on the example shown in \cref{fig:Merkle-tree-jumps}:
\begin{itemize}
\item First,
  we need an algorithm
  (described in detail in \cref{sssec:gen-m-alg})
  for generating the Merkle tree
  shown in \cref{fig:Merkle-tree-jumps}.
  It takes the sequence of
  words (the tokens) as input
  and outputs the label of the root node $u_{0,0}$.
  The tokens are given to the algorithm by the
  signing or quoting algorithms.
  The tree computation is somewhat trivial but important
  since different algorithms might result in different tree shapes,
  but the same shape is required for signing, quoting, and verifying.
  In our case, we require the tree to be heap-shaped.
\item The signing algorithm (described in \cref{sssec:signer-alg})
  extracts the sequence of tokens from the message,
  applies the algorithm for generating a Merkle tree,
  and signs the label of the root node.
\item The quoting algorithm (described in \cref{sssec:quote-alg})
  computes the nodes on the verification path of a quote
  along with their labels.
  Referring back to \cref{fig:Merkle-tree-jumps}, this starts with
  identifying the nodes highlighted in blue, if the quote was the
  subsequence \textquote{The quick}, and in red,
  if the quote was \textquote{jumps}. This algorithm also generates the
  Merkle tree and extracts the information needed to verify the quote
  such as the location of the quote in the original message
  as well as its length.
\item The verifying algorithm (described in \cref{sssec:verif-q}) is given the
  quote and the quoted signature, which consists of the location of its tokens
  within the original message, the length of the original message (number of
  tokens), the labels of the nodes on the corresponding verification path, and
  the digital signature for the label of the root hash.
  From this it calculates the label of the root node in the Merkle tree
  corresponding to the full sequence and verifies that the given signature is
  indeed a valid signature for this value.
  If the quote was \textquote{The quick}, this corresponds to calculating the
  labels of $u_{3,0},u_{3,1},u_{2,0},u_{1,0}$, and finally $u_{0,0}$, which
  would then be verified.
\end{itemize}
The following sections describe these algorithms in detail,
beginning with the computation of the Merkle tree,
since this operation is required
for both signing and quoting.

\subsubsection{Generating Merkle Trees}\label{sssec:gen-m-alg}

To generate a Merkle tree for a sequence $S$ of tokens, we define 
\texttt{CreateMerkleTree($S$)}
as a recursive function. We use $\ell$ to indicate
the number of tokens in $S$.
\begin{itemize}
\item If $S$ consists of just one token, then create a new node with the token as its
  label. Let this be the sole child of a new node $u$, with the hash of the token
  as its label. Return $u$.
\item Otherwise, create a new node $u$.
  \begin{itemize}
  \item If $\ell$ is a power of two, then let $u$'s left child be the node returned by
    recursively calling \texttt{CreateMerkleTree()} on the first $\ell/2$ tokens
    of $S$, and let $u$'s right child to be the node returned by recursively calling
    the \texttt{CreateMerkleTree()} on the last $\ell/2$ tokens of $S$.
  \item If $\ell$ is closer to $2^{\lceil \log_2\ell \rceil}$ than
    $2^{\lfloor \log_2\ell \rfloor}$, the tree rooted at $u$'s left child will be
    full and contain $2^{\lfloor \log_2\ell \rfloor}$ tokens. Hence, let $u$'s left
    child be the node returned by recursively calling
    \texttt{CreateMerkleTree()} on the first $2^{\lfloor \log_2\ell \rfloor}$ tokens
    of $S$, and $u$'s right child be the node returned by recursively calling
    \texttt{CreateMerkleTree()} on the remaining $\ell - 2^{\lfloor \log_2\ell \rfloor}$ tokens of $S$.
  \item If $\ell$ is closer to $2^{\lfloor \log_2\ell \rfloor}$ than
    $2^{\lceil \log_2\ell \rceil}$, the tree rooted at $u$'s right child will be
    complete, and contain $2^{\lfloor \log_2\ell \rfloor - 1}$ tokens. Hence, let
    $u$'s right child be the node returned by recursively calling
    this function on the last $2^{\lfloor \log_2\ell \rfloor-1}$ tokens
    of $S$, and $u$'s left child be the node returned by recursively calling
    this function
    on the remaining
    $\ell - 2^{\lfloor \log_2\ell \rfloor-1}$ tokens of $S$.
  \end{itemize}
\item Set $u$'s label to be the hash of the concatenation of the labels of $u$'s
  children, i.e., $u$.label = hash($u.\text{left.label} \concat
  u.\text{right.label}$).\footnote{Note that we are omitting that we have to
    mask the values of tokens and the labels of internal nodes in different ways
    before hashing them, in order to avoid a trivial collision attack, as
    discussed in \cref{ssec:mt-qs-security}. }
\item Return $u$.
\end{itemize}
This function returns the root of the Merkle tree corresponding to~$S$.
The Signer signs the label of the root
to create the digital signature for the message corresponding to $S$.


\subsubsection{Signing a Message}\label{sssec:signer-alg}
Using \texttt{CreateMerkleTree()}, signing a message is straightforward.
\begin{itemize}
\item Turn the message $m$ into a token sequence $S$.
  The Quoters and Verifiers need to be able
  to obtain the same tokens for a given message or quote.
  How this can be achieved 
  depends on the specific application and use-case;
  see~\cref{ssec:choices} for a brief discussion. 
\item Generate the Merkle tree for $S$
  using \texttt{CreateMerkleTree()}.
  Denote the label of the root of the Merkle
  tree by $u$.
\item Sign $u$ using a classical signature algorithm
  to obtain the quotable signature for the message $m$.
\end{itemize}

\subsubsection{Quoting a Message}\label{sssec:quote-alg}
To obtain a quote $Q$ for a subsequence of the token sequence $S$,
the Quoter does the following:
\begin{itemize}
\item Extract the token sequence from the message and generate the
  Merkle tree using \texttt{CreateMerkleTree()}. 
\item Add a flag to each internal node in the created Merkle tree
  that indicates if the label of the node needs to be provided in the
  signature for the quote. Initially, set each flag to~\texttt{delete}, indicating that
  they are not needed.
\item For each token \emph{in the quote}, process each node on its root-token
  path as described below (start at the node corresponding to the token and,
  after processing that node, continue to its parent, stopping after finishing
  with a child of the root). Note that when processing a later token, nodes
  on its root-token path may no longer have their flag set to \texttt{delete} if they have already
  been processed on another root-token path.
  \begin{itemize}
  \item If the node's flag is \texttt{delete},
    set its sibling's flag to \texttt{required}, indicating that its label is needed (unless
    this flag is later changed to \texttt{implicit}). Note that this node and its sibling could
    both correspond to tokens in the quote, in which case, when the sibling
    is processed, both nodes will have their flags set to \texttt{implicit}.
  \item If the node's flag is \texttt{required},
    set the flags of the node and its sibling to \texttt{implicit}, indicating
    that their labels can be calculated from information that is already
    included.
    Then move on to the next token; the rest of the verification path for this
    token has already been considered, as part of the verification path for a
    previously processed token.
  \end{itemize}
\item Extract the hashes that the signature needs to provide by
  performing an inorder traversal of the Merkle tree, adding the label of any
  node with its flag set to \texttt{required} to a list of provided hashes.
\item Create the signature for the quote
  as the signature for the root of the Merkle tree, the list of
  hashes generated in the previous step, the number of tokens in the
  original message, and the indices of the quoted tokens.
\end{itemize}

\begin{remark}
Note that we have made the assumption that the Quoter is quoting directly from
a message and not quoting from a quote. However, one can straightforwardly
combine the latter parts of this algorithm with parts of the algorithm described
in \cref{sssec:verif-q} to obtain this functionality.
\end{remark}

\subsubsection{Verifying a Quote}\label{sssec:verif-q}
Given a quote and a signature for the quote, consisting of the signature for the
root of the Merkle tree, a list of required hashes, the length of the original
sequence, and the indices of the quoted tokens, the Verifier can verify the quote as
follows:
\begin{itemize}
\item Create a heap-shaped tree with as many leaves as there were tokens
  in the original sequence. Let all the nodes be unlabeled.
\item Add a flag to each node in the tree, initially setting each flag to
  \texttt{delete}.
\item For each token \emph{in the quote}, work upwards on the root-token path
  corresponding to the token. For each node (except the root), do the following:
  \begin{itemize}
  \item If the node's flag is \texttt{delete}, set its sibling's flag to \texttt{required}.
  \item If the node's flag is \texttt{required}, set the flags of the node and
    its sibling to \texttt{implicit} and move on to the next token.
  \end{itemize}
\item Perform an inorder traversal of the tree. When encountering a node with
  its flag set to \texttt{required}, label it with the next hash in the list of required hashes. 
\item For each of the leaves corresponding to tokens in the quote, label them with the
  hash of that token.
\item The remaining labels, including the root's, can now be calculated using a
  straightforward recursive function: Starting from the root, calculate its
  label from the labels of its children, calling recursively on any unlabeled
  children.
\item Verify the calculated root hash, with respect to the signature for the
  root hash, included in the signature for the quote. If this verification is
  successful, the quote has been verified.
\end{itemize}

This only covers verifying the authenticity of the quote, and additional
information could be made clearly available.
This information could, for example, include if the quote is contiguous or where
tokens from the original message are missing, where they were located in the
message, and application-specific information.

\subsection{Application-Specific Choices}\label{ssec:choices}
When instantiating quotable signatures for a concrete use-case or application,
one of the choices to make is what to use as tokens.
In our examples, we have used words as tokens, which could be a natural choice
for some applications, but there are many other ways to tokenize a message.
This is considered further in \cref{sec:QSandFN}.

For the Signer's algorithm in \cref{sssec:signer-alg}, there is a choice
to be made as to what classical digital signature scheme is used to sign the
root's label. Here, suggestions could be to follow either one of schemes from the Digital Signature
Standard (DSS)~\cite{NIST-DSS}, or, in the interest of long-term security,
a post-quantum signature scheme such as~\cite{DBLP:journals/tches/DucasKLLSSS18,falcon,sphincs+v3.1}.

A natural optimization would be to change the Merkle trees to use tokens as
leaves instead of hashes of tokens. This would reduce the number of
hash calculations needed to construct a Merkle tree by about a half.
One can in some situations take this slightly further.
If the combined size of the tokens of two leaf children of a node is no
longer than a hash value, then we could use the concatenation of the two
tokens instead of a hash value for their parent.
Naturally, this could be continued recursively.
}
\section{Quotable Signatures and Fake News}\label{sec:QSandFN}
In the introduction, we argued that the current approach to mitigating the
effects of fake news, focusing on flagging problematic content, is not
sufficient. As mentioned, one supplementary approach could be to
bolster authentic content by authenticating the source of quotes, for example
on social media, and the literature gives reason to believe this could have
an impact. This approach could be implemented using a quotable
signature scheme. Here, the message that is the original source of a
quote would be an article and the creator or distributor of the article
(a news agency, for instance) would act as the Signer, the one sharing the quote as
the Quoter, and the one verifying the quote as the Verifier.
For this approach to be
effective, it would need to be widely adopted, both by news media and by users
sharing and reading quotes from articles. We make the following observations on
these problems.

Regarding the news media, there is wide interest in supporting initiatives to
combat fake news, see for example~\cite{C2PA}. Additionally, from our
discussions with a news media company,\footnote{Specifically, we talked with the
  editor in charge of the platforms and the editor in charge of the digital
  editorial office at a large
  media company that produces multiple newspapers for different regional areas,
  in both paper and digital versions.} it is apparent that the current workflow
employed by modern media companies is already highly automated, and it appears
that it should be quite simple to integrate a process by which, when an article
is published (or updated), it is automatically signed with the media company's
public key.
Regarding user adoption,
there is the challenge of getting a sufficiently large proportion of users using
the tool, but one would also have to teach users what a quote being
authenticated means, i.e., that the source and integrity of the quote has been
assessed, but not its truthfulness or the quality of its source, for example.

If news media and social media integrate this approach into their websites, our
\ShortLong{methods}{algorithms} can be employed without any explicit user awareness.
With such an integration, when a user copies a quote from
a signed article, a signature for the quote is automatically generated, and an
element including both quote and text is put into the clipboard, together with
the plain text quote (in practise, this would be a \texttt{text/html}
element and a \texttt{text/plain} element). When the user then pastes the quote,
a website supporting signatures will use the clipboard element with a signature~\cite{W3C_ClipboardAPI}.
One challenge with this approach is that the
verification is now performed by the websites, rather than a browser extension,
for example. Thus, the user has to trust the website to perform the
authentication correctly. 

An essential choice is how to divide text into
tokens, since any subsequence of the tokens is an allowable quote.
Natural choices could be by word, sentence, or paragraph. As a more involved
choice, one could also define the tokens at a per-token basis, and simply mark
the tokens in the HTML code. A variation of this would be to have a default
setting, but to allow the Signer to decide how to split the article into tokens
when signing.
As a variant, one could also consider using content extraction policies, as
in~\cite{DBLP:conf/icisc/SteinfeldBZ01}, so the Signer can specify which
subsequences of tokens are allowable quotes. A media company might want to disallow quotes of noncontiguous segments, for example, or disallow including only
parts of a sentence containing a negative, such as \textquote{not},
\textquote{neither}, or \textquote{never}. Such restrictions could be handled
efficiently using regular expressions.

We are implementing a prototype,\footnote{To be made available at \url{https://serfurth.dk/research/archive/}} separated into two
parts: a library that can be used by media companies to sign their articles and
a browser extension that allows users to quote with signatures and to verify
signatures for quotes. The library contains
implementations of the \ShortLong{proposed methods}{relevant algorithms from \cref{ssec:design-of-ss}} that each media
companies can integrate into their publishing workflow.
The browser extension modifies websites
such that text (both full articles and quotes) with verified signatures is
shown to be signed, and allows the user to make quotes from the signed text that
include a signature for the quote. The browser extension also allows
the user to get more information from the signature for a quote, e.g., who signed it,
when it was signed, an indication of where text was removed, and a link to the original article.

One could further extend the system with different labels, depending on the quality of the source of a quote.
For example,
many countries have press councils enforcing press ethics,
which includes providing correct information, e.g., by researching sufficiently
and publishing errata when needed.
Hence, it may make sense to mark quotes from articles written by news media
certified as following press ethics and rulings of a national press council.
One could even go so far as to authenticate only signatures signed by such
sources.

 To make a difference in the future, media companies and users on social media need
 to adopt these quotable signatures.
 To have the best effect,
 social media platforms should directly support quotable signatures
 and the required extension should be natively integrated into browsers.

\section{Future Work}
\label{sec:futureWork}

With this paper, we have extended the theory on quotable signatures and
presented an application of quotable signatures as a supplementary approach to
mitigating the effect of fake news. 

Further work on quotable signatures could include using methods
similar to the ones employed in \cite{DBLP:conf/pkc/HulsingR016}
and~\cite{sphincs+v3.1} to remove the requirement that the used hash function be
collision-resistant, and thereby remedy a vulnerability against multi-target
attacks against hash functions. Additionally, variants of quotable signatures
optimized for different types of media should be developed and
compared. Our current variant is in some sense optimized for cases where one
will often wish to quote something contiguous in one dimension, such as
text. If, instead, the goal is to crop an image, one would end up with
a \textquote{quote} that is contiguous in two dimensions. We have not yet explored how
to handle this case effectively.
Finally, as discussed in \cref{sec:QSandFN}, different policies for dividing text
into tokens could be studied.

A natural next step towards using quotable signatures to combat misinformation
would be to verify the effectiveness of 
the proposed method experimentally. In particular, the effects of using quotable
signatures for verifying news shared on social media and elsewhere need to be
investigated. A suggestion for a first study could be to investigate if the use
of quotable signatures improves participants' ability to recall from which news
brand a story originated, which was an issue identified
in~\cite{Newsbrandattribution18}. Additional studies along the lines of
\cite{doi:10.1126/sciadv.abl3844}, investigating the effects on the quality of
the news diet of participants, would also be of interest.


\hypersetup{ citecolor=blue, linkcolor=blue, urlcolor=blue}
\printbibliography

\appendix
\end{document}